\documentclass[letterpaper, 11pt, conference]{ieeeconf}  % Comment this line out
\IEEEoverridecommandlockouts                              % This command is only

\usepackage{float}
\usepackage{mathptmx} 
\usepackage{times} 
\usepackage{amsmath} 
\usepackage{amsbsy} 
\usepackage{amssymb}
\usepackage[printwatermark]{xwatermark}
\usepackage{mathrsfs}
\usepackage{comment}
\usepackage{subcaption}
\usepackage[export]{adjustbox}
\usepackage{tikz}
\usetikzlibrary{external,positioning,decorations.pathreplacing,shapes,arrows,patterns}
\usepackage{pgfplots}
\usepackage{graphicx}
\usepackage{pstool}
\usepackage[latin1]{inputenc}
\usetikzlibrary{arrows,shapes}
\usepackage{xifthen}
\usepackage{epic}
\usepackage{caption}
\usepackage{epstopdf}
\usepackage[bookmarks=false]{hyperref}

\newtheorem{thm}{\bf{Theorem}}

\newtheorem{lem}[thm]{\bf {Lemma}}
\newtheorem{prop}[thm]{\bf {Proposition}}

\newcommand{\mmse}{\mathsf{mmse}}

\newcommand*{\QEDA}{\hfill\ensuremath{\square}}

\tikzstyle{int}=[draw, fill=blue!10, minimum height = 1cm, minimum width=1.5cm,thick ]
\tikzstyle{sint}=[draw, fill=blue!10, minimum height = 0.5cm, minimum width=0.8cm,thick ]
\tikzstyle{sum}=[circle, fill=blue!10, draw=black,line width=1pt,minimum size = 0.5cm, thick ]
\tikzstyle{ssum}=[circle, fill=blue!10,draw=black,line width=1pt,minimum size = 0.1cm]
\tikzstyle{int1}=[draw, fill=blue!10, minimum height = 0.5cm, minimum width=1cm,thick ]
\tikzstyle{enc}=[draw, fill=blue!10, minimum height = 2.7cm, minimum width=1cm,thick ]
\tikzstyle{int}=[draw, fill=blue!10, minimum height = 1cm, minimum width=1.5cm,thick ]

\title{\LARGE \bf
Lossy Compression of Decimated Gaussian Random Walks
} 

\author{Georgia Murray, Alon Kipnis, and Andrea J. Goldsmith\\
Department of Electrical Engineering\\
Stanford University, Stanford, CA, USA%
}
\begin{document}

\begin{NoHyper}

\maketitle
\thispagestyle{empty}
\pagestyle{empty}

\begin{abstract}
    We consider the problem of estimating a Gaussian random walk from a lossy compression of its decimated version. Hence, the encoder operates on the decimated random walk, and the decoder estimates the original random walk from its encoded version under a mean squared error (MSE) criterion. It is well-known that the minimal distortion in this problem is attained by an \emph{estimate-and-compress} (EC) source coding strategy, in which the encoder first estimates the original random walk and then compresses this estimate subject to the bit constraint. In this work, we derive a closed-form expression for this minimal distortion as a function of the bitrate and the decimation factor. Next, we consider a \emph{compress-and-estimate} (CE) source coding scheme, in which the encoder first compresses the decimated sequence subject to an MSE criterion (with respect to the decimated sequence), and the original random walk is estimated only at the decoder. 
    We evaluate the distortion under CE in a closed form and show that there exists a non-zero gap between the distortion under the two schemes. This difference in performance illustrates the importance of having the decimation factor at the encoder. 
    %Numerical evaluations reveal that the distortion under the CE scheme approaches that of the EC scheme at high bitrates $R$, while a substantial performance gap exists between the performances of the two schemes at low rates. 
\end{abstract}


\section{Introduction \label{sec:intro}}

\begin{figure} [b]
\begin{center}
\begin{tikzpicture} [auto,>=latex]
 \node at (0,0) (source) {$X^N$} ;
 \node[int1, right of = source , node distance = 2cm] (decimator) {$\downarrow M$};
 \node[int1, right of = decimator , node distance = 3cm, text width = 0.7cm] (interp) {Est};
 \node[coordinate, right of = interp, node distance = 1cm] (right_up) {};  
  \node[int1, below of = interp, node distance = 1.5cm] (enc1) {$\mathrm{Enc}$};  
\draw[->,line width = 2pt] (source) -- (decimator); 
\draw[->,line width = 2pt] (decimator) -- node[above] {$Y^{N_M}$} (interp); 
\draw[line width = 2pt] (interp) -- node[above, xshift = 0.5cm] {$\widetilde{X}^N$} (right_up); 
\draw[->,line width = 2pt] (right_up) |- (enc1); 
\node[int1, below of = decimator, node distance = 1.5cm ] (dec) {Dec};
\node[below of = source, node distance = 1.5cm] (dest) {$\widehat{X}^N$};
\draw[<->,dashed] (source) -- node[right] {$\|x-\widehat{x}\|^2$} (dest);
\draw[->,line width = 2pt] (enc1) -- node[above, xshift = 0cm] (mes1) {$ \left\{0,1\right\}^{\lfloor N R \rfloor} $} (dec);   
\draw[->, line width=2pt] (dec) -- (dest);
\end{tikzpicture}
\end{center}
\caption{\label{fig:ECsys} 
Estimate-and-compress (EC) strategy. The decimated data is first used to obtain an optimal MSE estimate of the source sequence $X^N$. The encoder then describes this estimate to the decoder using $NR$ bits. 
}
\end{figure}
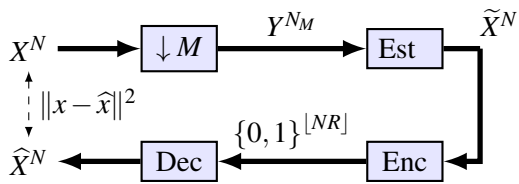

\begin{figure} [b]
\begin{center}
\begin{tikzpicture} [auto,>=latex]
 \node at (0,0) (source) {$X^N$} ;
 \node[int1, right of = source , node distance = 2cm] (decimator) {$\downarrow M$};
 \node[int1, right of = decimator , node distance = 3cm ] (interp) {Enc};
 \node[coordinate, right of = interp, node distance = 1cm] (right_up) {};  
  \node[int1, below of = interp, node distance = 1.5cm] (enc1) {$\mathrm{Dec}$};  
\draw[->,line width = 2pt] (source) -- (decimator); 
\draw[->,line width = 2pt] (decimator) -- node[above] {$Y^{N_M}$} (interp); 
\draw[line width = 2pt] (interp) -- node[above, xshift = 0.5cm] {$ \left\{0,1\right\}^{\lfloor N R \rfloor} $} (right_up); 
\draw[->,line width = 2pt] (right_up) |- (enc1); 

%\draw[dashed] (interp)+(-0.7,0) -- +(-0.7,0.6) -- node[above] {$\mathrm{Enc}$} +(0.7,0.6) -- +(0.7,-2) -- +(-0.7,-2) -- +(-0.7,0);

\node[int1, below of = decimator, node distance = 1.5cm, text width = 0.7cm ] (dec) {Est};

%\node[int1, below of = dec, node distance = 1cm] (dec1) {$\mathrm{Dec}_Z$};

%\draw[->,dashed,line width = 2pt] (enc1)+(-1,0) |- (dec1);   
\node[below of = source, node distance = 1.5cm] (dest) {$\widehat{X}^N$};

%\draw[<->,dashed] (source) -- node[left] {$D$} (dest);

%\node[below of = dest, node distance = 1cm] (dest1) {$\widehat{Z}^{NM/L}$};

\draw[<->,dashed] (source) -- node[right] {$\|x-\widehat{x}\|^2$} (dest);

\draw[->,line width = 2pt] (enc1) -- node[above] {$\widehat Y^{N_M}$}(dec);   

%\draw[->, line width=2pt] (dec1) -- (dest1);
\draw[->, line width=2pt] (dec) -- (dest);
\end{tikzpicture}
\end{center}
\caption{\label{fig:CEsys} Compress-and-estimate (CE) strategy. The observed sequence $Y^{N_M}$ is encoded so as to minimize the MSE between $Y^{N_M}$ and $\widehat{Y}^{N_M}$. The source sequence $X^N$ is estimated from the output of the decoder.}
\end{figure}
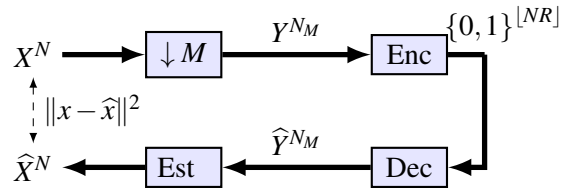

Consider the situation in which one is interested in compressing or transmitting a data sequence $X^N = (X_1,\ldots,X_N)$, generated by a source $X$, but can only access its factor $M$ decimated version
\[
Y_n = X_{Mn},\quad n = 1,\ldots, N/M.
\]

Assuming a compression or communication rate of $R$ bits per symbol to describe $X$, the encoder has at most $RN$ bits to represent $Y^{N_M}$, where $N_M\triangleq N/M$. Without loss of generality, we assume here and throughout the paper that $N_M$ is an integer. The problem of finding the $NR$ bit representation that minimizes the MSE with respect to $X^N$ is known as the \emph{indirect} (or remote) source coding problem. Classical results in source coding \cite{1057738,1054469,1056251, Kipnis2014} show that optimal compression is achieved by an \emph{estimate-and-compress} (EC) strategy: the encoder first computes $\widetilde{X}^N$, the optimal estimate of $X^N$ from its decimated version $Y^{N_M}$, and then compresses $\widetilde{X}^N$ under a MSE criterion subject to the bit constraint (Fig.~\ref{fig:ECsys}). The resultant expected MSE is called the indirect DRF of $X^N$ given $Y^{N_M}$, and we denote it here by $D_{EC}(M,R)$. As $N$ goes to infinity, $D_{EC}(M,R)$ is a function of only the bitrate $R$ and the decimation factor $M$ \cite{Kipnis2014}.
\par
%However, in order to achieve the MSE estimate $\widetilde{X}^N$, EC requires that the decimation factor $M$ be known at the encoder, which may not be practical in all applications.
%
This work focuses on the fact that, in many situations, the encoder cannot compute $\widetilde{X}^N$. This may be the result of an unknown decimation factor $M$ or a lack of computing resources. When the encoder is unable to estimate $X^N$ prior to encoding, a different scheme known as \emph{compress-and-estimate} is often employed \cite{KipnisRini2017}. In this scheme, depicted in Fig. \ref{fig:CEsys}, the observed decimated sequence $Y^{N_M}$ is encoded in an optimal manner subject to an MSE criterion. The original sequence $X^N$ is estimated at the decoder from the compressed version of $Y^{N_M}$. The distortion under this scheme is denoted as $D_{CE}(M,R)$ and provides an upper bound for the distortion under EC. That is, it bounds from above the minimal distortion in the indirect source coding problem of $X^N$ given $Y^{N_M}$ when the decimation factor $M$ is unknown at the encoder.\\

In this paper we focus on the case where the source $X$ is a standard Gaussian random walk, defined as
\begin{equation}
    \label{eq:X_def}
    X_n = \sum_{i=1}^n W_i,\quad n=1,2,\ldots,
\end{equation}
where $W_1,\ldots,W_n$ are standard normal and independent of each other. The process \eqref{eq:X_def} arises as the uniform samples of the Wiener process \cite{KipnisWiener}, or the \emph{discrete-time} Wiener process. Applications of the random walk are many, ranging from diffusion models in physics to option pricing in financial mathematics \cite{weiss1994aspects}. \par 
The main contributions of this paper are closed form expressions of the distortion functions $D_{EC}(M,R)$ and $D_{CE}(M,R)$ for the Gaussian random walk $X$ defined by \eqref{eq:X_def}. These expressions fully characterize the fundamental limit arising from representing a Gaussian random walk by quantizing its samples at a fixed bitrate, as is necessary when, for example, transmitting them over a rate-limited link. Moreover, we show that for any decimation factor $M>1$ and bitrate $R>0$, the distortion under CE is strictly sub-optimal compared to the minimal distortion achieved by EC. That is, the scheme that encodes the decimated sequence $Y$ so as to recover $Y$ with minimal distortion, attains distortion in recovering $X$ that is strictly larger than the scheme that encodes $Y$ so as to recover $X$ with minimal distortion. This result illustrates that, in problems involving inference from lossy compressed information, the optimal lossy compression procedure depends on the end inference problem. As a result, ad-hoc lossy compression techniques that do not take into account the final inference procedure are necessarily sub-optimal. Nevertheless, our results reveal that the difference between $D_{EC}(M,R)$ and $D_{CE}(M,R)$ is relatively small, and may be insignificant in many applications.\\
%
%
%Our contributions in analyzing the lossy compression according to the EC and CE schemes for this process are as follows:
%\begin{enumerate}
%\item We derive a closed-form expression for the distortion under EC, describing the minimal distortion in recovering $X^N$ from any rate $R$ representation of $Y^{N_M}$ under a MSE criterion.
%\item We characterize the distortion under CE using  cross correlation between error terms following from the information expression for the DRF of $X^N$. 
%We denote this expression by $D_{CE}(R,M)$.
%We evaluate this expression numerically and compare it to the distortion under EC. Consequently, we derive conclusions on the performance loss in using the sub-optimal CE scheme compared to the optimal EC. 
%\end{enumerate}

This paper is organized as follows. In Sec.~\ref{sec:problem} we define our general source coding problem and the EC and CE schemes. 
In Sec.~\ref{sec:background} we review relevant known results with respect to our general indirect source coding problem. In Sec.~\ref{sec:main} we characterize the distortion under the EC and CE schemes. In Sec.~\ref{sec:numerical_results} we evaluate the resulting distortion expressions numerically and derive conclusions regarding the loss of performance when the decimation factor is unknown to the encoder. Finally, we provide concluding remarks and discuss future work in Sec.~\ref{sec:conclusions}.

\begin{comment}
\begin{figure*}
    \centering
    \textbf{Eigenvalues of decimated Random Gaussian Walk}
    \includegraphics[width=\linewidth]{eval_xtild.png}[h]
    \caption{Comparison of eigenvalues for covariance matrix of $\widetilde{X}^N$ as derived in \ref{eq:tild_lam} with numerical evaluation in Matlab for $N = 1000$. The x-axis shows the varying number of eigenvalues dependent upon $M$, as explained in \ref{eq:tild_lam}. The $M$ values chosen correspond to those used for direct comparison between $D_X(R)$ and $D_{EC}(R)$, and relative difference is defined as $|\frac{Numerical - Analytical}{Analytical}|$.}
    \label{fig:eval_xtild}
\end{figure*}
\end{comment}

\begin{figure} [b]
\begin{center}
\begin{tikzpicture} [auto,>=latex]
 \node at (0,0) (source) {$X^N$} ;
 \node[int1, right of = source , node distance = 1.5cm] (decimator) {$\downarrow M$};
 \node[int1, right of = decimator , node distance = 2cm] (enc) {$\mathrm{Enc}$};
  
\draw[->,line width = 2pt] (source) -- (decimator); 
\draw[->,line width = 2pt] (decimator) -- node[above] {$Y^{N_M}$} (enc); 

\node[int1, right of = enc, node distance = 2.5cm ] (dec) {$\mathrm{Dec}$};

\node[right of = dec, node distance = 1.5cm] (dest) {$\widehat{X}^N$};

\draw[->,line width = 2pt] (enc) -- node[above, xshift = 0cm] (mes1) {$ \left\{0,1\right\}^{\lfloor N R \rfloor} $} (dec);   
\draw[->, line width=2pt] (dec) -- (dest);
\end{tikzpicture}
\end{center}
\caption{\label{fig:problem_setting} 
Decimation and source coding setting. 
}
\end{figure}
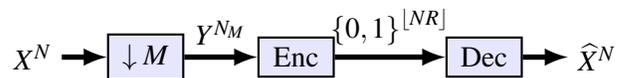

\section{Problem Formulation \label{sec:problem}}
We consider the source coding problem described in Figs.~\ref{fig:ECsys} and \ref{fig:CEsys}. In this problem, the standard Gaussian random walk $X$ of \eqref{eq:X_def} is decimated by a factor $M$ to yield the process $Y$. Note that $Y$ is also a Gaussian random walk, though with variance $M$ rather than unit variance. For a time horizon $N$, the vector $Y^{N_M}$ is encoded using the encoder 
\begin{equation}
f : \mathbb R^{N_M} \rightarrow \left\{0,1\right\}^{NR} \label{eq:encoder_def}. 
\end{equation}
The decoder $g :  \left\{0,1\right\}^{NR} \rightarrow \mathbb R^N$ receives the binary word at the output of the encoder and provides a reconstruction sequence $\widehat{X}^N = g(f(Y^{N_M}))$. The distortion is defined as the normalized MSE between $X^N$ and $\widehat{X}^N$:
\[
D_{f,g} \triangleq \frac{1}{N} \sum_{n=1}^N \mathbb E \left(X_n - \widehat{X}_n \right)^2.
\]
Our goal in this paper is to characterize $D_{f,g}$ in the limit as $N \rightarrow \infty$ under two different of encoder and decoder designs: the optimal design which has knowledge of the decimation factor at the encoder and decoder, and a suboptimal design which does not have this knowledge at the encoder.

%In this paper we are concerned with the rate-distortion trade-off for two compression-recovery systems for the decimated Wiener process, as described in Section \ref{sec:intro}. Rate is defined as the number of bits per source symbol; thus for a source $X^N$ at rate $R$, $NR$ bits are allocated. Our distortion metric is defined as the means squared difference between the original source sequence and the final recovered sequence. That is, for a realization $x^N$ of $X^N$, the distortion is given by
%\[
%\left\| x^N - \widehat{x}^N \right\|^2 = \frac{1}{N} \sum_{n=1}^N \left(x_n-\widehat{x}_n \right)^2. 
%\]

\subsection{Optimal Source Coding via EC}
The optimal source coding performance with respect to the source coding problem of Fig.~\ref{fig:problem_setting} is defined as the minimum of $D_{f,g}$ under all pairs of encoders and decoders. Since the encoder in this problem has no direct access to the signal $X^N$ it aims to accurately represent, the characterization of this minimal distortion is an indirect source coding problem \cite[Ch. 3.5]{berger1971rate}. Classical results in source coding show that the minimum of $D_{f,g}$ is attained by the EC strategy illustrated in Fig.~\ref{fig:ECsys}. That is, the encoder first estimates $X^N$ from the observed signal $Y^{N_M}$, and then compresses this estimated version in an optimal manner as in classical source coding \cite{1057738,1056251,1054469}. 
For this reason, we set
\[
D_{EC}(M,R) \triangleq \liminf_{N\rightarrow \infty} \inf_{f,g} D_{f,g}.
\]
$D_{EC}(M,R)$ is called the indirect DRF of the process $X$ given the process $Y$ as it describes the asymptotic optimal performance in indirect source coding.

\subsection{CE Source Coding}
The CE scheme is defined by a particular sequence of encoders that generally differ from the optimal one used in EC. Specifically, the encoder $f_{CE}$ in CE is a minimum distance encoder with respect to a set of $2^{NR}$ codewords drawn from the distribution that attains the DRF of the Gaussian vector $Y^{N_M}$ at bitrate not exceeding $R$. The decoder $g_{CE}$ receives the index of the codeword $\hat{y}^{N_M}$ nearest to the input sequence and outputs $\hat{x}^N$, obtained by linearly interpolating $\hat{y}^{N_M}$ as in
\begin{equation}
    \label{eq:interp_ce}
    \widehat{x}_n = \frac{M-n}{M} \widehat{y}_{n^-} + \frac{n}{M}\widehat{y}_{n^+},\quad n=1,\ldots,N,
\end{equation}
where $n^+ = \lceil \frac{n}{M} \rceil$ and $n^- = \lfloor \frac{n}{M} \rfloor$. Note that in the CE setting, although the encoding is optimal with respect to $Y^{N_M}$, it is \textit{not} necessarily optimal with respect to $X^N$. However, the decimation factor $M$ is not used by the encoder in CE, and hence this scheme may be useful when $M$ is unknown. \par
A distortion $D$ is said to be achievable under CE if there exists a sequence of encoders of the form $f_{CE}$ such that $D_{f,g}$ converges to $D$ as $N\rightarrow \infty$. We denote by $D_{CE}(M,R)$ the infimum over all achievable distortions under CE.

\section{Background \label{sec:background}}
In this section we review relevant known results for encoding the Gaussian random walk $X^N$ of \eqref{eq:X_def}. \par
Since $X^N$ is Gaussian and Markovian, the minimal MSE (MMSE) estimate of $X^N$ from $Y^{N_M}$ is simply the interpolation of the decimated version. That is
\[
    \widetilde{X}_n \triangleq \mathbb E\left[X_n | Y^{N_M} \right] 
    = \frac{M-n}{M} Y_{n^-} + \frac{n}{M} Y_{n^+},\quad n=1,2,\ldots, 
\]
where $n^+ = \lceil \frac{n}{M} \rceil$ and $n^- = \lfloor \frac{n}{M} \rfloor$. The resulting MMSE, which we denote by $\mmse(M)$, is given by
\begin{align*}
\mmse(M) \triangleq \frac{1}{N} \sum_{n=1}^N \mathbb E \left(X_n - \widetilde{X}_n \right)^2  =\frac{M-M^{-1}}{6}. 
\end{align*}
 Note that due to the properties of conditional expectation, for any encoder $
f : \mathbb R^{N_M} \rightarrow \left\{1,\ldots,2^{NR} \right\}$ 
we have
\begin{equation}
\begin{aligned}
\mmse \left(X^N| f(Y^{N_M}) \right) & \\ =  \mmse&\left(M\right) + \mmse( \widetilde{X}^N | f(Y^{N_M})). \label{eq:decomp}
\end{aligned}
\end{equation}
Therefore, $\mmse(M)$ is a trivial lower bound to the functions $D_{EC}(M,R)$ and $D_{CE}(M,R)$. Moreover, as explained in \cite{1054469}, it follows from \eqref{eq:decomp} that the minimal distortion in estimating $X^N$ from any $NR$-bit representation of $Y^{M/N}$ is attained by the optimal encoding of $\widetilde{X}^N$ subject to this bit constraint. Hence, the EC scheme, in which the encoder first estimates $\widetilde{X}^N$ and then encodes it, is optimal. \\

Another, trivial lower bound to $D_{EC}(M,R)$ and $D_{CE}(M,R)$ is given by the (standard) DRF of the process $X$. This DRF is defined as the limit infimum as $N\rightarrow \infty$ of the normalized distortion of the Gaussian vector $X^N$. The latter is given via Kolmogorov's expression \cite{1056823} 
\begin{subequations} \label{eq:drf}
\begin{align}
    D_{X^N}(\theta) &= \frac{1}{N} \sum_{k=1}^N \min[\theta, \lambda_k]\\
    R_{X^N}(\theta) &= \frac{1}{2N} \sum_{k=1}^N \max[0, \log (\lambda_k/\theta)],
\end{align}
\end{subequations}
where $\lambda_k$'s are the eigenvalues of the covariance matrix $\Sigma_{X^N}$ of $X^N$. In our case of $X^N$ as a standard Gaussian random walk, Berger \cite{berger1970information} showed that
\begin{equation}
    \lambda_k = \Big [2 \sin \Big (\frac{2k-1}{2N+1} \frac{\pi}{2}\Big ) \Big ]^{-2}, \qquad k = 0, ..., N-1,
\end{equation}
 and concluded, upon taking the limit in \eqref{eq:drf}, that
\begin{subequations}
\label{eq:drf_int}
\begin{align}
    D_{X}(R_\theta) & = \int_0^1 \min \left\{\theta, S(\phi) \right\} d\phi \\
    R_\theta & =  \frac{1}{2} \int_0^1 \max \{0,\log \left(S(\phi)/\theta \right) \}d\phi,
    \end{align}
\end{subequations}
where $S(\phi) = \left(2 \sin \left( \pi \phi/2 \right) \right)^{-2}$ is the asymptotic density of the eigenvalues of $\Sigma_{X^N}$.

\section{Distortion under EC and CE \label{sec:main}}
We now derive our main results by characterizing the distortion under EC and CE in recovering the random walk $X$ from its decimated version $Y$. \\

\subsection{Estimate-and-Compress}
From the definition of $D_{EC}(M,R)$ and the decomposition \eqref{eq:decomp}, it follows that 
\begin{align}
    D_{EC}(M,R) &= \mmse(M) + \liminf_{N\rightarrow \infty } \inf_{f} \mmse(X^N|f(Y^{N_M})) \nonumber \\
    & = \mmse(M) +  D_{\widetilde{X}}(R), 
\label{eq:decomp_ce}
\end{align}
where $D_{\widetilde{X}}(R)$ is the DRF of the process $\widetilde{X}$. Therefore, characterizing the distortion in EC is obtained by solving a source coding problem with respect to $\widetilde{X}$. Now the process $B$ defined as $B_n \triangleq \left( X_n - \widetilde{X}_n \right)$ returns to zero at least every $M$ steps and has average variance equals to $\mmse(M)$. Hence the variance of $\widetilde{X}_n = X_n - B_n$ increases at the same rate as the variance of $X_n$, and Berger's coding theorem for $X_n$ \cite{berger1970information} can be applied to $\widetilde{X}_n$. Therefore, the DRF of $\widetilde{X}$ is given by the limiting expression for the DRF of the Gaussian vector $\widetilde{X}^N$, using Kolmogorov's expression \eqref{eq:drf} leading to the following result:
\begin{thm} \label{thm:EC_DRF}
Let 
\[
\widetilde{S}(\phi) \triangleq \left(2 \sin \left(\phi \pi/2 \right) \right)^{-2} - \frac{1-M^{-2}}{6}. 
\]
Then the indirect DRF of the random walk $X$ given its factor $M$ decimated version $Y$ equals
\begin{subequations} \label{eq:DEC_integral}
\begin{align}
    D_{EC}(M,R_\theta) &= \mmse(M) +  M\int_0^1 \min\{\theta, \widetilde{S}(\phi) \} d\phi \\
    R_\theta  &= \frac{1}{2 M} \int_0^1 \max \left\{0, \log \left(\widetilde{S}(\phi) / \theta \right) \right\} d\phi,
\end{align}
\end{subequations}
\end{thm}
\begin{proof}
We show in the Appendix that the $N_M$ non-zero eigenvalues of the covariance matrix of $\widetilde{X}^N$ are given by
\begin{equation} \label{eq:tild_lam}
    \begin{split}
        \widetilde{\lambda}_k(M) &= M^2\Big[2\sin \Big( \frac{(2k-1)M}{2N+1}\frac{\pi}{2} \Big) \Big]^{-2} - \frac{M^2-1}{6},\\
        &\qquad k = 0...N_M-1
    \end{split}
\end{equation}
Substituting \eqref{eq:tild_lam} into \eqref{eq:drf} we have
\begin{subequations} 
\begin{align}
    \label{eq:DEC}
    D_{\widetilde{X}^N}(R_\theta) &=  \frac{1}{M} \frac{M}{N} \sum_{k=1}^{N_M} \min[\theta, \widetilde{\lambda}_k] \\
    R_\theta  &= \frac{1}{M} \frac{M}{2N} \sum_{k=1}^{N_M} \max[0, \log (\widetilde{\lambda}_k/\theta)],
\end{align}
\end{subequations}
where the $\widetilde{\lambda}_k$'s are given by \eqref{eq:tild_lam}. Taking the limit in \eqref{eq:DEC} as $N\rightarrow \infty$ with $kM/N\rightarrow \phi \in (0,1)$ and $\theta' = \theta/M$ leads to the integral representation for $D_{\widetilde{X}}(R)$. Finally,  \eqref{eq:DEC_integral} is obtained by adding the MMSE term to $D_{\widetilde{X}}(R)$.
\end{proof}

\subsection{Compress-and-Estimate}
We now consider the compress-and-estimate scheme. As in EC, we begin from the decomposition in \eqref{eq:decomp}. However, instead of using the optimal encoder that attains the DRF of $\widetilde{X}$, we use the encoder $f_{CE}$ that maps $Y^{N_M}$ to one of $2^{NR}$ possible sequences $\widehat{y}^{N_M}(1),\ldots,\widehat{y}^{N_M}(2^{NR})$. By linearity of $\widetilde{X}^N$ in $Y^{N_M}$, we have that the MMSE estimate of $X^N$ from $f_{EC}(Y^{N_M})$ is given by the interpolation \eqref{eq:interp_ce}, hence
\[
\mmse\left(X^N | \widehat{Y}^{N_M} \right) = \mmse \left(X^N | f_{CE} \left(Y^{N_M} \right) \right),
\]
where $\widehat{Y}^{N_M} = g_{EC} \left(f_{CE}( Y^{N_M})\right)$. Therefore, in order to derive $D_{CE}$ via \eqref{eq:decomp}, it is left to characterize the term $\mmse (\widetilde{X}^N | f_{CE}(Y^{N_M}))$. 
Note that unlike in EC, this term does not describe a distortion under optimal encoding, since while optimal encoding was performed, it was performed with respect to $Y^{N_M}$ rather that $\widetilde{X}^N$. Therefore, we characterize $\mmse (\widetilde{X}^N | f_{CE}(Y^{N_M}))$ by expressing it in terms of the error in encoding $Y^{N_M}$ with respect to the CE codebook:
\[
\epsilon \triangleq Y^{N_M} - \hat{Y}\left(f_{CE}(Y^{N_M}) \right),
\]
This connection is achieved by the following lemma:
\begin{lem} \label{lem:mse_connection}
For any $N$, $M$, and encoder $f$ we have:
\begin{equation} \label{eq:DCE_oper}
\begin{split}
    & \mmse(\widetilde{X}^N|\widehat Y^{N_M} ) = \frac{2M^2+1}{3NM} \sum_{n=1}^{N_M+1} \mathbb E[\epsilon_n^2] \\
    & + \frac{M^2-1}{3NM} \sum_{n=1}^{N_M} \mathbb E[\epsilon_n \epsilon_{n+1}]  - \frac{2M^2+3M+1}{6NM} \mathbb E[\epsilon_{N_M}^2].
    % & \mmse(\widetilde{X}^N|f(Y^{N_M}) ) = \frac{2M^2+1}{3M^2} \frac{N+M}{N} \mmse(Y^{N_M+1}| f(Y^{N_M}) ) \\
    % & +  \frac{1}{NM} \left( \sum_{n=0}^{N_M-1}  \frac{M^2-1}{3} \mathbb E[\epsilon_n \epsilon_{n+1}]  - \frac{2M^2+3M+1}{6} \mathbb E[\epsilon_{N_M}^2]  \right). 
\end{split}
\end{equation}
\end{lem}
The proof of Lem.~\ref{lem:mse_connection} can be found in the Appendix. \\

\begin{comment}
The following upper bound to $D_{CE}(M,R)$ is a direct consequence of Lem.~\ref{lem:mse_connection}:
\begin{prop} \label{prop:upper_bound}
Let $f_y$ describe a sequence of encoders such that
\[
\lim_{N\rightarrow \infty} \mmse \left( Y^{N_M} | f_y(Y^{N_M}) \right) = D_Y(M R),
\]
where $D_Y(M R) = MD_X(M R)$ is obtained from \eqref{eq:drf_int}. Then 
\begin{equation} \label{eq:prop_1}
\liminf_{N\rightarrow \infty} \mmse \left(X^N | f_y(Y^{N_M}) \right) \leq \mmse(M) + MD_X(M R).
\end{equation}
In particular 
\begin{equation}
    \label{eq:DCE_ub}
D_{CE}(R) \leq \mmse(M) + MD_X(M R).
\end{equation}
\end{prop}
\begin{proof}
Since $\mathbb E[\epsilon_n \epsilon_{n+1}] \leq \frac{1}{2}\mathbb E \left[\epsilon_n^2 + \epsilon_{n+1}^2 \right]$, it follows from \eqref{eq:DCE_oper} that
\[
    \mmse(\widetilde{X}^N | f(Y^{N_M}) \leq  \frac{N+M}{N} \mmse\left(Y^{N_M+1}| f_y(Y^{N_M}) \right).
\]
\eqref{eq:prop_1} follows since $f_y$ is such that $\mmse\left( Y^N| f_y(Y^N)\right) \rightarrow D_Y(\bar{R})$ and since $\bar{R} = MR$ is the number of bits per symbol allotted to $f_y$ to encode $Y^{N_M}$. 
\end{proof}
\end{comment}

%From its definition, the encoder $f_{CE}$ satisfies the condition in Prop.~\ref{prop:upper_bound} and hence leads to the upper bound \eqref{eq:DCE_ub}.
Using Lem.~\ref{lem:mse_connection} with the encoder  $f_{CE}$, we obtain a closed-form expression for $D_{CE}(M,R)$, as per the following theorem: 
\begin{thm} \label{thm:CE_DRF}
For any decimation factor $M$ and bitrate $R$, the infimum over all acheivable distortions using the CE scheme is given by
\begin{subequations}
\label{eq:CE_DRF}
\begin{align}
    D_{CE}(M,R_\theta) & = \mmse(M) + \frac{2M^2+1}{3M} \int_0^1 \min \left\{S(\phi),\theta\right\} d\phi \nonumber \\
    & + \frac{M^2-1}{3M} \int_0^1 \min \left\{S(\phi),\theta \right\} \cos(\pi\phi) d\phi \label{eq:CE_DRF_D}, \\
    R_\theta  = & \frac{1}{2M} \int_0^1 \max \left[0, \log \left(S(\phi) /\theta \right) \right] d\phi, \label{eq:CE_DRF_R}
\end{align}
\end{subequations}
\end{thm}
where $S(\phi) = \left(2 \sin \left( \pi \phi/2 \right) \right)^{-2}$, as in \eqref{eq:drf_int}. \\

\begin{proof}
Only a sketch of the proof is provided here. The full proof can be found in the Appendix. In view of \eqref{eq:decomp}, it is enough to show that $\mmse \left(\widetilde{X}^N|f_{CE}(Y^{N_M}) \right)$ converges to the water-filling part in \eqref{eq:CE_DRF}. Using Lem.~\ref{lem:mse_connection} with $f_{CE}$ implies that the first term in the RHS of \eqref{eq:DCE_oper} converges to 
$\frac{2M^2+1}{3M}D_X(MR)$, and leads to the first term in \eqref{eq:CE_DRF_D}. In order to evaluate the term $\mathbb E[\epsilon_n \epsilon_{n+1}]$ in \eqref{eq:DCE_oper}, we consider the properties of the encoder $f_{CE}$. The joint distribution of the two sequences $Y^{N_M}$ and $\hat{Y}^{N_M}$, at the input and output of the encoder, respectively, behaves as if both sequences were drawn from the joint $P^*_{Y^{N_M},\hat{Y}^{N_M}}$ that attains the DRF of the vector $Y^{N_M}$ \cite{gallager1968information, kontoyiannis2006mismatched}. In our case, this distribution is defined by a Gaussian channel $P_{\hat{Y}^{N_M},Y^{N_M}}$. Therefore, by setting $\hat{\epsilon} \triangleq Y - \widehat{Y}$, we conclude that 
\[
 \mathbb E \left[ \epsilon_n \epsilon_{n+1}\right] = \mathbb E \left[ \hat{\epsilon}_n\hat{\epsilon}_{n+1}\right]
 %= \left( U \Lambda_\theta U^H \right)_{n,n+1} 
 = \sum_{k=1}^{N_M} u_k[n] u_k[n+1], 
\]
where $u_k$'s are the eigenvectors of covarience matrix $\Sigma_{Y^{N_M}}$, given in \cite{berger1970information}. The behavior of the last term in the limit $N\rightarrow \infty$ leads to the second term in \eqref{eq:CE_DRF_D}.
\end{proof}

\begin{comment}
\begin{figure*}[h]
    \centering
    \textbf{$D_{EC}$ vs. $D_X$\\Distortion increase through downsampling}
    \begin{subfigure}[t]{0.5\textwidth}
        \centering
        \includegraphics[width=\linewidth]{inc_dist_ec_abs.png}
        \label{fig:DEC_dist_abs}
        \caption{Absolute}
    \end{subfigure}%
    ~ 
    \begin{subfigure}[t]{0.5\textwidth}
        \centering
        \includegraphics[width=\linewidth]{inc_dist_ec_rel.png}
        \label{fig:DEC_dist_rel}
        \caption{Relative}
    \end{subfigure}
    \caption{Distortion increase due to downsampling. Absolute $\triangleq D_{EC}-D_X$. Relative $\triangleq \frac{D_{EC} - D_X}{D_X}$}
    \label{fig:DEC_dist}
\end{figure*}

\begin{figure*}[h]
    \centering
    \textbf{$D_{CE}$ vs. $D_{EC}$\\Distortion increase through non-optimal compression}
    \begin{subfigure}[t]{0.5\textwidth}
        \centering
        \includegraphics[width=\linewidth]{inc_dist_ce_abs.png}
        \caption{Absolute}
    \end{subfigure}%
    ~ 
    \begin{subfigure}[t]{0.5\textwidth}
        \centering
        \includegraphics[width=\linewidth]{inc_dist_ce_rel.png}
        \caption{Relative}
    \end{subfigure}
    \caption{Distortion increase due to non-optimal encoding. Absolute $\triangleq D_{EC} - D_{CE}$. Relative $\triangleq \frac{D_{EC} - D_{CE}}{D_{EC}}$}
    \label{fig:DCE_dist}
\end{figure*}

\end{comment}

\section{Analysis and Interpretations}\label{sec:numerical_results}

Since the parameter $\theta$ obscures the direct dependency of $D_{EC}$ and $D_{CE}$ on $R$, we will consider the conditions under which we can eliminate the parameter $\theta$.  We will then numerically analyze the dual dependency of $D_{EC}$ and $D_{CE}$ and $M$ and $R$.

\subsection{High Rate Characterizations}
When the number of bits per decimated symbol $MR$ is large, $\theta$ can be eliminated from \eqref{eq:DEC_integral} and \eqref{eq:CE_DRF}, leading to single-line expressions for $D_{EC}(M,R)$ and $D_{CE}(M,R)$. This leads to the following proposition.
\begin{prop} \mbox{}
\label{prop:properties}
\begin{itemize}
    \item[(i)] For $R M \geq 1$,
    \begin{equation}
        \label{eq:DCE_high_rate}
        D_{CE}(M,R) = \mmse(M) +  \frac{2M^2+1}{3M} 2^{-2MR}
    \end{equation}

\item [(ii)] For $ M R \geq 
\log2 \left[1 + 3/\sqrt{3+\frac{6}{M^2}} \right] \geq 1$,
%   $M R\geq \log \left( 1+  \left( 1-\frac{2}{3}(1-M^{-2}) \right)^{-1/2} \right)  $, 
    \begin{align}
        D_{EC}(M,R) & =\mmse(M)  \label{eq:DEC_high_rate} \\
        &\quad + \frac{1+\left(2+\sqrt{3+\frac{6}{M^2}}\right) M^2}{6 M} 2^{-2 M R} \nonumber
    \end{align}
\end{itemize}
\end{prop}
~\\
The proof of Prop.~\ref{prop:properties} can be found in the Appendix. Note that, as expected, \eqref{eq:DEC_high_rate} and \eqref{eq:DCE_high_rate} reduce to $D_X(R)$ for $M=1$. It follows from Prop.~\ref{prop:properties} that
\begin{equation}
    \label{eq:diff}
    \begin{aligned}
D_{CE}(M,R) -& D_{EC}(M,R)  \\ &=  \frac{1+\left(2-\sqrt{3+\frac{6}{M^2}}\right) M^2}{6 M} 2^{-2 M R},
\end{aligned}
\end{equation}
whenever $M R\geq 
\log2 \left[1 + 3/\sqrt{3+\frac{6}{M^2}} \right]$.

\subsection{Discussion}

\begin{figure} [t]
    \centering
    %\textbf{Distortion vs. Bitrate\\$M=100$}
    \includegraphics[width=0.9\linewidth]{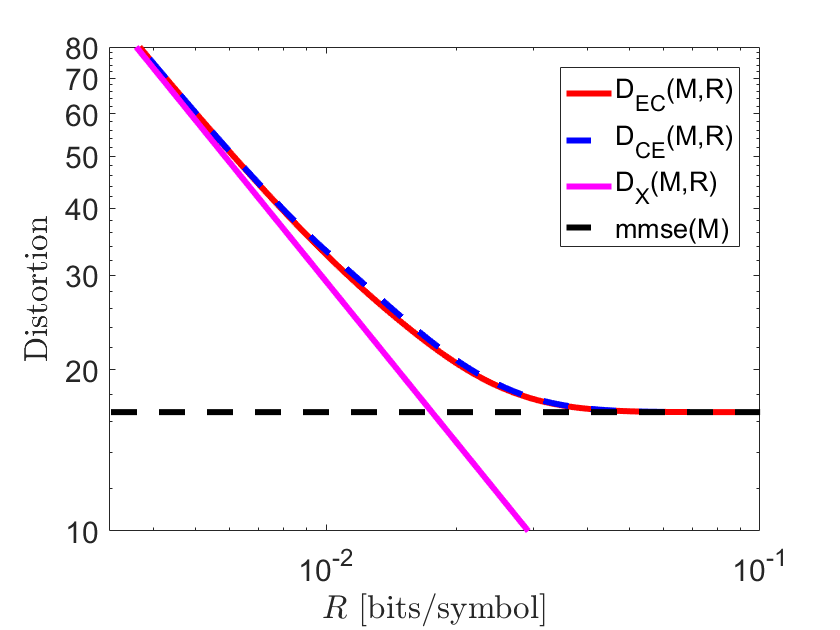}
    \includegraphics[width=0.9\linewidth]{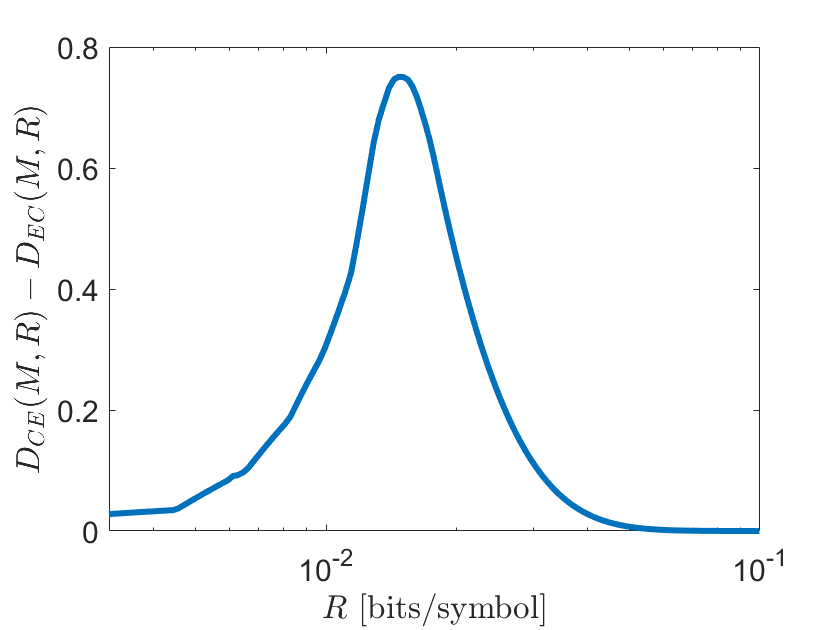}
    \caption{Distortion as a function of the bitrate $R$ for a fixed downsampling factor $M = 100$. The bottom figure shows the difference $D_{CE}(M,R) - D_{EC}(M,R)$. }
    \label{fig:D_vs_R}
\end{figure} 

\begin{figure} [t]
    \centering
    %\textbf{Distortion vs. Downsampling Factor\\$R=0.1~\text{bits/symbol}$}
    \includegraphics[width=0.9\linewidth]{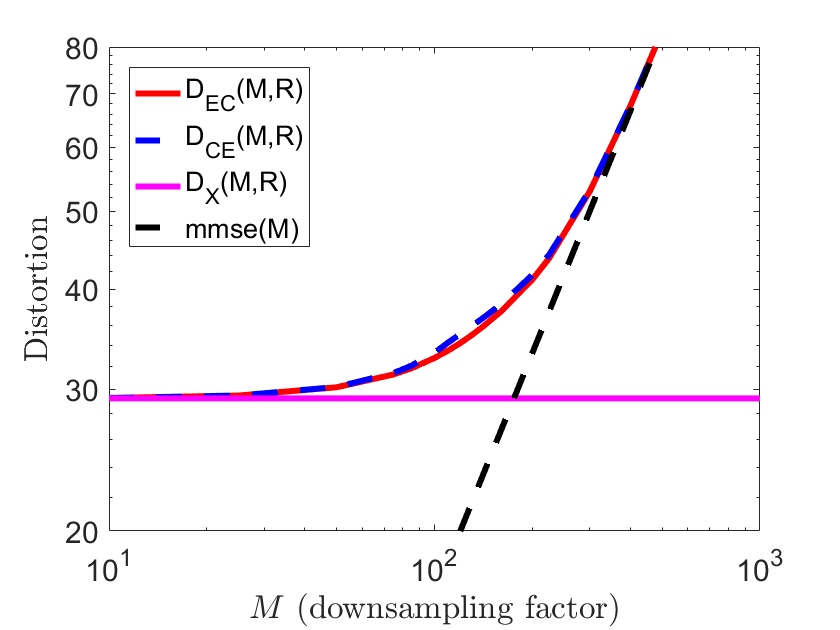}
    \includegraphics[width=0.9\linewidth]{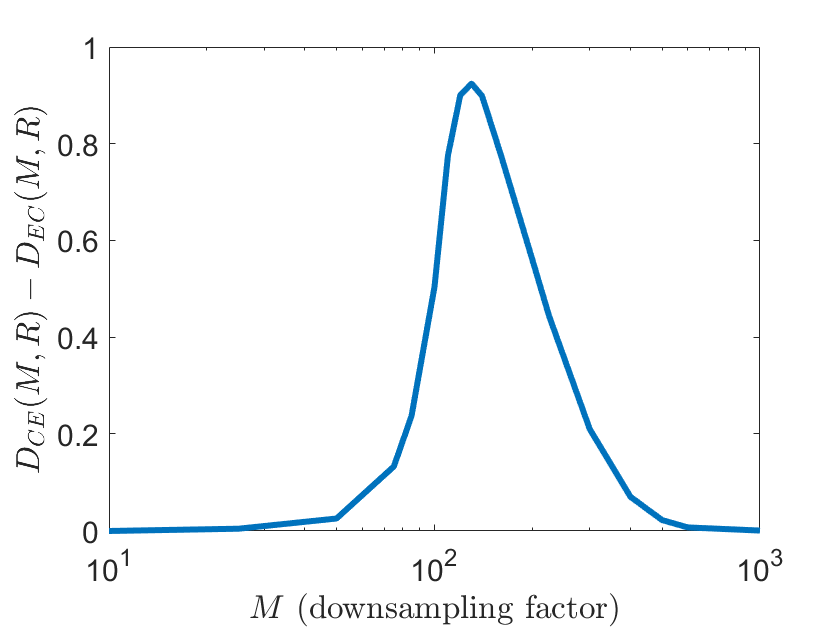}
    \caption{Distortion as a function of downsampling factor $M$ for a fixed $R = 0.01$ bits per symbol of $X$. The bottom figure shows $D_{CE}(M,R) - D_{EC}(M,R)$. 
    }
    \label{fig:D_vs_M}
\end{figure}

Figs.~\ref{fig:D_vs_R} and \ref{fig:D_vs_M} show the distortion expressions as functions of the bitrate $R$ and the decimation factor $M$, respectively. We see that both $D_{EC}(M,R)$ and $D_{CE}(M,R)$ are bounded from below by $D_X(R)$ and $\mmse(M)$, representing the minimal distortion only due to lossy compression and decimation, respectively. Further, both $D_{EC}(M,R)$ and $D_{CE}(M,R)$ approach the bounds in the two extremes of no decimation ($M=1)$ and infinite bitrate ($R\rightarrow \infty$). That is, for low values of $R$ compared to $M$, the distortion under both schemes is dominated by the error due to lossy compression, whereas the distortion is dominated by interpolation error when $R$ is large compared to $M$. \par
%
% \begin{figure} [h]
%     \centering
%     %\textbf{Distortion ratio for CE vs EC}
%     \includegraphics[width=0.9\linewidth]{D_R_Diff.png}
%     \caption{$D_{CE}(M,R) - D_{EC}(M,R)$ as a function of bitrate $R$ bits per symbol of $X$ for a fixed downsampling factor $M = 0.1$.}
%     \label{fig:D_R_Diff}
% \end{figure}
%
The bottom of Figs.~\ref{fig:D_vs_R} and \ref{fig:D_vs_M} illustrate the performance gap in using CE compared to EC, given by \eqref{eq:diff} for sufficiently large $MR$. This gap is maximal in the transition region between rate-dominant and decimation-dominant distortion. As can be seen in Fig.~\ref{fig:D_vs_R}, although the performance gap is positive for any $M$ and $R$, it is relatively small. For example, when $M=100$, the maximal value of \eqref{eq:diff}, i.e. the performance loss from using CE instead of the optimal EC, is $2.7\%$, and thus CE may be used as a near approximation to optimal performance when EC is impractical or, due to an ignorance of $M$ at the encoder, impossible.

%Fig.~\ref{fig:D_vs_M} illustrates the compression schemes dependency on the decimation factor $M$. Similarly to the trends discussed for $R$, both schemes converge to the two error floors for large and small values of $M$. Fig. \ref{fig:D_M_Diff} shows that $D_{CE}$ converges to $D_{EC}$ quickly as $M$ increases.

% \begin{figure} [h]
%     \centering
%     %\textbf{Distortion ratio for CE vs EC}
%     \includegraphics[width=0.9\linewidth]{D_M_Diff.png}
%     \caption{$D_{CE}(M,R) - D_{EC}(M,R)$ as a function of downsampling factor $M$ for a fixed $R = 0.1$ bits per symbol of $X$.}
%     \label{fig:D_M_Diff}
% \end{figure} 

\section{Conclusions \label{sec:conclusions}}

We have derived a closed-form expression for the minimal distortion in recovering a Gaussian random walk from a finite-bit representation of its decimated version. 
This expression quantifies the behavior of the minimal distortion subject to decimation and lossy compression. This expression also confirms the following expected behavior: convergence to the standard DRF of the random walk at low coding rates where encoding error dominates; an interpolation error floor at high coding rates where decimation error dominates; and increased degradation with increasing decimation. \par
In addition, we considered the distortion in recovering the random walk under a CE scheme. In this scheme, the encoding of the decimated process is done with respect to a random codebook derived from its rate-distortion achieving distribution. That is, a codebook that is designed to attain the DRF of the decimated process, rather than the original (not decimated) random walk. In particular, the encoder in this scheme is not informed of the decimation factor. 
%DRF of the decimated random Gaussian walk under the optimal estimate-and-compress scheme ($D_{EC}$). We further provided a closed-form expression for the distortion under the CE scheme ($D_{EC}$), which can be achieved when the decimation factor $M$ is unknown at the encoder. \par
 The comparison between the two distortion expressions  provides the excess distortion as a result of using the sub-optimal CE encoding. In particular, it provides a distortion bound on the price of an unknown decimation factor at the encoder. We show that this price is small and need be considered only for a narrow band of $R$ and $M$ values where neither source of distortion, i.e. neither the bit constraint nor the decimation, have become dominant distortion factors.\par
 
%A particularly interesting avenue for future research arise when the decimation factor is estimated or speculated (perhaps incorrectly). In this situation, the encoder can interpolate based on this approximated decimation factor. This scenario leads to a hybrid EC-CE system, where partial interpolation is performed pre-encoding and final interpolation is performed post-decoding. We expect that such a system could mitigate the increased distortion of the CE scheme without requiring exact knowledge of $M$ at the encoder. Fig.~\ref{fig:ICE_shceme} depicts such a scheme.

\section*{Acknowledgments}
This research was supported in part by the NSF Center for Science of Information (CSoI) under grant CCF-0939370.

\bibliographystyle{IEEEtran}
\bibliography{IEEEabrv,sampling}

\begin{thebibliography}{10}
\providecommand{\url}[1]{#1}
\csname url@samestyle\endcsname
\providecommand{\newblock}{\relax}
\providecommand{\bibinfo}[2]{#2}
\providecommand{\BIBentrySTDinterwordspacing}{\spaceskip=0pt\relax}
\providecommand{\BIBentryALTinterwordstretchfactor}{4}
\providecommand{\BIBentryALTinterwordspacing}{\spaceskip=\fontdimen2\font plus
\BIBentryALTinterwordstretchfactor\fontdimen3\font minus
  \fontdimen4\font\relax}
\providecommand{\BIBforeignlanguage}[2]{{%
\expandafter\ifx\csname l@#1\endcsname\relax
\typeout{** WARNING: IEEEtran.bst: No hyphenation pattern has been}%
\typeout{** loaded for the language `#1'. Using the pattern for}%
\typeout{** the default language instead.}%
\else
\language=\csname l@#1\endcsname
\fi
#2}}
\providecommand{\BIBdecl}{\relax}
\BIBdecl

\bibitem{1057738}
R.~Dobrushin and B.~Tsybakov, ``Information transmission with additional
  noise,'' vol.~8, no.~5, pp. 293--304, 1962.

\bibitem{1054469}
J.~Wolf and J.~Ziv, ``Transmission of noisy information to a noisy receiver
  with minimum distortion,'' \emph{{IEEE} Trans. Inf. Theory}, vol.~16, no.~4,
  pp. 406--411, 1970.

\bibitem{1056251}
H.~Witsenhausen, ``Indirect rate distortion problems,'' \emph{{IEEE} Trans.
  Inf. Theory}, vol.~26, no.~5, pp. 518--521, 1980.

\bibitem{Kipnis2014}
A.~Kipnis, A.~J. Goldsmith, Y.~C. Eldar, and T.~Weissman, ``Distortion rate
  function of sub-{N}yquist sampled {G}aussian sources,'' \emph{{IEEE} Trans.
  Inf. Theory}, vol.~62, no.~1, pp. 401--429, Jan 2016.

\bibitem{KipnisRini2017}
\BIBentryALTinterwordspacing
A.~Kipnis, S.~Rini, and A.~J. Goldsmith, ``Compress and estimate in
  multiterminal source coding,'' 2017, unpublished. [Online]. Available:
  \url{https://arxiv.org/abs/1602.02201}
\BIBentrySTDinterwordspacing

\bibitem{KipnisWiener}
\BIBentryALTinterwordspacing
A.~Kipnis, A.~J. Goldsmith, and Y.~C. Eldar, ``The distortion-rate function of
  sampled {W}iener processes,'' \emph{CoRR}, vol. abs/1608.04679, 2016.
  [Online]. Available: \url{http://arxiv.org/abs/1608.04679}
\BIBentrySTDinterwordspacing

\bibitem{weiss1994aspects}
G.~Weiss, \emph{Aspects and Applications of the Random Walk," Random Materials
  and Processes, Series Eds. H. Stanley and E. Guyon}.\hskip 1em plus 0.5em
  minus 0.4em\relax North Holland, 1994.

\bibitem{berger1971rate}
T.~Berger, \emph{Rate-distortion theory: A mathematical basis for data
  compression}.\hskip 1em plus 0.5em minus 0.4em\relax Englewood Cliffs, NJ:
  Prentice-Hall, 1971.

\bibitem{1056823}
A.~Kolmogorov, ``On the {S}hannon theory of information transmission in the
  case of continuous signals,'' vol.~2, no.~4, pp. 102--108, December 1956.

\bibitem{berger1970information}
T.~Berger, ``Information rates of {W}iener processes,'' \emph{{IEEE} Trans.
  Inf. Theory}, vol.~16, no.~2, pp. 134--139, 1970.

\bibitem{gallager1968information}
R.~Gallager, \emph{Information theory and reliable communication}.\hskip 1em
  plus 0.5em minus 0.4em\relax Wiley, 1968.

\bibitem{kontoyiannis2006mismatched}
I.~Kontoyiannis and R.~Zamir, ``Mismatched codebooks and the role of entropy
  coding in lossy data compression,'' \emph{{IEEE} Trans. Inf. Theory},
  vol.~52, no.~5, pp. 1922--1938, 2006.

\end{thebibliography}

\newpage

\onecolumn
\section*{Appendix}

\subsection{Eigenvalues of $\Sigma_{\widetilde{X}^N}$}
Here we sketch our derivation of the eigenvalues for the covariance matrix of the down-sampled and interpolated sequence $\widetilde X$, as given in \eqref{eq:tild_lam} and used in the characterization of the $D_{EC}$ in \eqref{eq:DEC}. We follow the same proof approach as used by Berger for the non-decimated random Gaussian walk in \cite{berger1970information}.

Let $\widetilde \Sigma_{{X}^N}$ be the covariance matrix of the interpolated process $\widetilde{X}^N$ with eigenvalues $\widetilde \lambda$ and eigenvectors $\widetilde{f}$. Under the assumption of unit variance, this gives us

 \begin{displaymath}
 \begin{split}
 \begin{aligned}
 \widetilde\lambda \widetilde f(i) &= \sum_{j = 0}^{i^- - 1} j \widetilde f(j) + \sum_{j = i^-}^{i^+ - 1}  \widetilde f(j)\big(\frac{M - j}{M} i^- + \frac{j}{M} i\big) + i \sum_{j = i^+}^{N-1} \widetilde f(j) \\
 \widetilde\lambda[\widetilde f(i&+2) - 2\widetilde f(i+1) + \widetilde f(i)] = 0 \\
 \end{aligned}
 \end{split}
 \end{displaymath}

 From this we conclude that the eigenvectors $f$ must be piece-wise linear. However, since the interpolation and downsampling factors are the same, the boundary conditions remain the same as in the uninterpolated Wiener process case. In \cite{berger1970information}, they were shown as 
 \begin{displaymath}
 \widetilde f(0) = 0 \qquad \widetilde{\lambda}[\widetilde f(n)-\widetilde f(n-1)]=\widetilde f(n)
 \end{displaymath}

 Thus to satisfy both the piece-wise linearity as well as the boundary conditions, we consider piece-wise linear interpolations of the decimated eigenvectors of the original Wiener process. That is, if we let $\lambda$ and $f$ be the eigenvalues and eigenvectors of $\Sigma_X$, the covariance of the un-decimated sequence $X$, as derived in \cite{berger1970information},

 \begin{equation}\label{eq:f_xtild}
 \begin{split}
 %&\widetilde f_k(jM) = f_k(jM), \text{ where j is an integer} \\
 &\widetilde f_{k}(n) = B \left( \frac{n^+ - n}{M} f_k(n^-) + \frac{n - n^-}{M} f_k(n^+) \right) \\
 &\qquad \qquad f_{k}(n) = A \sin \left( \frac{2k - 1}{2N + 1} \pi n \right) \\
 &\widetilde f_{k}(n) = C \bigg(\frac{n^+-n}{M}\sin \Big[\Big(\frac{2k-1}{2N+1}\Big) \pi i^-\Big] +\frac{n - n^-}{M} \sin\Big[\Big(\frac{2k-1}{2N+1}\Big) \pi i^+\Big]\bigg)
 \end{split}
 \end{equation}

 \noindent where $n^+ = \lceil \frac{n}{M} \rceil M$ and $n^- = \lfloor \frac{n}{M} \rfloor M$, and $A,C$ are normalization constants. Since each element of $\widetilde{X}$ is the result of a linear combination of elements of the decimated sequence $Y$, the dimensionality of $\widetilde{X} \leq$ the dimentionality of $Y$, and thus \eqref{eq:f_xtild} only holds for $k \leq N_M$. For $k > N_M$, $\widetilde \lambda_k = 0$, and thus we are unconcerned with the corresponding eigenvectors.

 Now to determine the eigenvalues, let $p \triangleq jM$, where j is an integer.

 \begin{displaymath}
 \begin{split}
 \begin{aligned}
 &\widetilde f(p) = f(p) = A \sin \left( \frac{2k - 1}{2N + 1} \pi p \right) \\
 &\widetilde \lambda[\widetilde f(p+2M) - 2\widetilde f(p+M) + \widetilde f(p)] \\ &= - C \frac{1}{3} \sin(\tau \pi (p + M)) [(M^2 - 1) \cos(\tau \pi M) + 2M^2 + 1] \\
 &\qquad \qquad \tau \triangleq \frac{2k-1}{2N+1} \\
 &\widetilde \lambda_{k} = \frac{C}{A} \frac{(M^2 - 1) \cos(\tau \pi M) + 2M^2 + 1}{12 M^2 \sin^2(\tau \frac{\pi}{2} M)} \\
 &\widetilde \lambda_{k} = D \left( \frac{2\cos^2(\tau \frac{\pi}{2} M) + 1}{12 \sin^2(\tau \frac{\pi}{2} M)} + \frac{1}{6M^2} \right)\\
 &\widetilde{\lambda}_k = M^2\Big[2\sin \Big( \frac{(2k-1)M}{2N+1}\frac{\pi}{2} \Big) \Big]^{-2} - \frac{M^2-1}{6},\\
 &\qquad D, C, A = \text{normalization constants}
 \end{aligned}
 \end{split}
 \end{displaymath}
This proves Eq \ref{eq:tild_lam}.

\subsection{Proof of Lem.~\ref{lem:mse_connection}}

We here provide the complete proof for \eqref{eq:DCE_oper}.  Note that throughout the main paper, all sequences are treated as 1-indexed (that is, the first element is $X_1$).  For the simplicity of the calculation, the following derivation is done for a 0-indexed sequence (thus the first element is $X_0$).  However, by simply re-indexing the final result, we achieve \eqref{eq:DCE_oper}.

\begin{displaymath}
\begin{split}
    % \mmse(X | \widehat Y) &= \mmse(X | Y) + \mmse (\widetilde X | \widehat Y) \\
    % &\qquad \mmse(X|Y) = \text{Interpolation Error -- known} \\
    \mmse(\widetilde X | \widehat Y) &= \frac{1}{N} \sum_{n=0}^{N-1} \mathbb E \big[\widetilde X_n - \mathbb E [ \widetilde X_n | \widehat Y] \big]^2\\
    &= \frac{1}{N} \sum_{n=0}^{N_M-1} \sum_{m=nM}^{(n+1)M-1} \mathbb E \big[\widetilde X - \mathbb E [ \widetilde X | \widehat Y] \big]^2\\
    &= \frac{1}{N}\sum_{n=0}^{N_M-1} \sum_{m=nM}^{(n+1)M-1} \mathbb E \Big[\widetilde X - \frac{(n+1)M - m}{M} \widehat Y_{n} - \frac{m - nM}{M} \widehat Y_{n+1} \Big]^2 \\
    &\qquad \widehat Y_n = Y_n - \epsilon_n \\
    &= \frac{1}{N}\sum_{n=0}^{N_M-1} \sum_{m=nM}^{(n+1)M-1} \mathbb E \Big[\frac{(n+1)M - m}{M} \epsilon_{n} + \frac{m - nM}{M} \epsilon_{n+1} \Big]^2 \\
    &= \frac{1}{NM}\sum_{n=0}^{N_M-1} \frac{2M^2+3M+1}{6} \mathbb E[\epsilon_n^2] + \frac{M^2-1}{3} \mathbb E[\epsilon_n \epsilon_{n+1}] + \frac{2M^2-3M+1}{6} \mathbb E[\epsilon_{n+1}^2]\\
    &= \frac{1}{NM}\Bigg ( \sum_{n=0}^{N_M} \Big [ \frac{2M^2+1}{3} \mathbb E[\epsilon_n^2] \Big ] + \sum_{n=0}^{N_M-1} \Big [ \frac{M^2-1}{3} \mathbb E[\epsilon_n \epsilon_{n+1}]\Big] \\&\qquad \qquad - \frac{2M^2-3M+1}{6} \mathbb E[\epsilon_0^2] - \frac{2M^2+3M+1}{6} \mathbb E[\epsilon_{N_M}^2] \Bigg ) \\
    & \qquad \qquad \mathbb E [\epsilon_0^2] = 0 \\
    &= \frac{2M^2+1}{3NM}\sum_{n=0}^{N_M} \mathbb E[\epsilon_n^2] + \frac{M^2-1}{3NM} \sum_{n=0}^{N_M-1} \mathbb E[\epsilon_n \epsilon_{n+1}] - \frac{2M^2+3M+1}{6NM} \mathbb E[\epsilon_{N_M}^2]
    % &= \frac{2M^2+1}{3M^2} \frac{N+M}{N} \mmse(Y^{N_M+1}|\hat Y) +  \frac{1}{NM}\Bigg ( \sum_{n=0}^{N_M-1} \Big [ \frac{M^2-1}{3} \mathbb E[\epsilon_n \epsilon_{n+1}]\Big] - \frac{2M^2+3M+1}{6} \mathbb E[\epsilon_{N_M}^2] \Bigg ) \\
    % &\text{UPPER BOUND: } \mathbb E[\epsilon_n \epsilon_{n+1}] \leq \mathbb E \Big[ \frac{\epsilon_n^2 + \epsilon_{n+1}^2}{2} \Big ]
\end{split}
\end{displaymath}

Re-indexing from $1$ rather than $0$ yields \eqref{eq:DCE_oper}.

\subsection{Convergence of $\mathbb E[\hat{\epsilon}_n \hat{\epsilon}_{n+1}]$}

In order to evaluate the cross term in \eqref{eq:DCE_oper}, first note that
\begin{align*}
& \frac{1}{N_M} \sum_{n=1}^{N_M-1} u_k[n]u_{k}[n+1]  \\
& =  \frac{1}{N_M} \sum_{n=1}^{N_M-1} A_k^2 \sin \left(\frac{2k-1}{2N_M+1} \frac{\pi}{2}n \right) 
 \sin \left(\frac{2k-1}{2N_M+1} \frac{\pi}{2}(n+1) \right) \\
 & = \frac{1}{N_M} \sum_{n=1}^{N_M-1} A_k^2 \sin \left(\frac{2k-1}{2N_M+1} \pi n \right) 
 \sin \left(\frac{2k-1}{2N_M+1} \pi(n+1) \right) \\
 %& = \frac{A_k^2}{2N_M} \sum_{n=1}^{N_M-1}  \sin \left(\frac{2k-1}{2N_M+1} \pi \right) 
 %\sin \left(\frac{2k-1}{2N_M+1} \pi(2n+1) \right) \\
 & = \frac{A_k^2}{2} \cos \left(\frac{2k-1}{2N_M+1} \pi \right) + o(1),
\end{align*}
where $o(1)\overset{N\rightarrow \infty}{\longrightarrow} 0$, and this last transition is because 
\[
\sum_{n=1}^{N_M-1} \cos \left(\frac{2k-1}{2N+1}\pi(2n+1) \right)
\]
is bounded in $N_M$. Similarly, we have
\begin{align*}
1 & = \sum_{n=1}^{N_M} \left(u_k[n]\right)^2 = \sum_{n=1}^{N_M} A_k^2 \sin^2 \left(\frac{2k-1}{2N_M+1} \pi n \right) \\
& = \frac{A_k^2 N_M}{N_M} \sum_{n=1}^{N_M} \left(\frac{1}{2} - \frac{1}{2}\cos \left( 2 \frac{2k-1}{2N_M+1} \pi n \right) \right),
\end{align*}
and hence $A_k^2 N_M \overset{N\rightarrow \infty}{\longrightarrow} 2$. As a result
\begin{align}
    &\frac{1}{N_M} \sum_{n=1}^{N_M-1} \mathbb E [\hat{\epsilon}_n \hat{\epsilon}_{n+1}] \nonumber  \\
    & = \frac{1}{N_M} \sum_{n=1}^{N_M-1} \sum_{k=1}^{N_M} u_k[n]u_k[n+1]\min\{\theta,\lambda_k \} \nonumber  \\
    & = \frac{1}{N_M} \sum_{k=1}^{N_M} \min \{\theta, \lambda_k\} \sum_{n=1}^{N_M-1} u_k[n]u_k[n+1] \nonumber  \\
    & = o(1) + \frac{1}{N_M} \sum_{k=1}^{N_M} \frac{N_M A_k^2}{2} \min \{\theta, \lambda_k\} \cos\left(\frac{2k-1}{2N_M+1} \pi  \right) \label{eq:cross_term_proof}
\end{align}
We now take the limit in \eqref{eq:cross_term_proof} as $N\rightarrow \infty$ with $k/N_M = kM/N \rightarrow \phi \in [0,1]$. In this limit, the spectrum of $\Sigma_Y$ converges to $M S(\phi)$, where
\[
S(\phi) \triangleq \left[ 2 \sin \left(\frac{\phi \pi}{2} \right) \right]^{-2}. 
\]
Therefore, the sum in \eqref{eq:cross_term_proof} converges to 
\[
\int_0^1 \min \left\{ \theta, M S(\phi) \right\} \cos(\pi \phi ) d\phi = M \int_0^1 \min \left\{ \theta', S(\phi) \right\} \cos(\pi \phi ) d\phi, 
\]
where $\theta' = \theta /M$.

\subsection{Proof of Prop.~\ref{prop:properties}}

When $\theta \leq 1/4$, \eqref{eq:drf_int} reduces to 
\begin{equation}
\label{eq:drf_int_low_theta}
\begin{split}
    D_{X}(R_\theta) & = \theta \\
    R_\theta & =  \log\Big[\frac{1}{ \sqrt{\theta}}\Big]
\end{split}
\end{equation}
and hence $D_X(R) = 2^{-2R}$. Since $D_{CE}(M,R)$ depends on the same asymptotic eigenvalue density $S(\phi)$, we conclude from \eqref{eq:CE_DRF} that 
\begin{align*}
    D_{CE}(M,R) & = \frac{M-M^{-1}}{6} + \left( \frac{2M^2+1}{3M} \right) 2^{-2MR}.
\end{align*}

%Thus at high bitrates ($R$) or severe decimation ($M$), we expect the $D_{CE}$ to converge towards the interpolation error $\mmse(M) = \frac{M-M^{-1}}{6}$.\\

For the function $D_{EC}(M,R)$, the minimum of $\widetilde{S}(\phi)$ is
\[
 \frac{1}{2} - \frac{1-M^{-2}}{6} = \frac{1 + M^{-2}}{12},
\]
and for $\theta$ smaller than this value we have
\begin{align*}
    & D_{EC} = \frac{M-M^{-1}}{6} + M \theta \\
    & M R  = \frac{1}{2} \int_0^1 \log \left[ \big(2\sin(\phi\pi/2)\big)^{-2} - \frac{1-M^{-2}}{6}\right] d\phi - \frac{\log(\theta)}{2}, \\
    & = -1 + \log \left[1+\sqrt{1-4 \frac{1-M^{-2}}{6}} \right] - \frac{\log(\theta)}{2}.
\end{align*}
Eliminating $\theta$ from the last expression, leads to \eqref{eq:DEC_high_rate}
\[
D_{EC}(M,R) = \frac{M-M^{-1}}{6} + \left( \frac{1+\left(2+\sqrt{3+\frac{6}{M^2}}\right) M^2}{6 M} \right) 2^{-2MR},
\]
which holds whenever 
\begin{align*}
MR & \geq -1 + \log \left[1+\sqrt{1+4 \frac{1-M^{-2}}{6} } \right] - \log \left[\frac{1+M^{-2}}{12} \right] \\
& = \log \frac{ \sqrt{3}M+\sqrt{5M^2-2}} {\sqrt{2+M^2}}.
\end{align*} 

\subsection{Proof of Thm.~\ref{thm:CE_DRF} }
In view of \eqref{eq:decomp}, it is enough to show that $\mmse \left(\widetilde{X}^N|f_{CE}(Y^{N_M}) \right)$ converges to the water-filling part in \eqref{eq:CE_DRF}. Usign Lem.~\ref{lem:mse_connection} we have
\begin{align*}
\frac{2M^2+1}{3M^2} \frac{N+M}{N}\mmse \left( Y^{N_M}|f_{CE}(Y^{N_M}) \right)  \\
\overset{N\rightarrow \infty}{\longrightarrow} \frac{2M^2+1}{3M^2}D_Y(MR) = \frac{2M^2+1}{3M}D_X(MR),
\end{align*}
which, using \eqref{eq:drf_int} with $\theta'=\theta/M$, leads to \eqref{eq:CE_DRF_R} and to the first term in \eqref{eq:CE_DRF_D}. In order evaluate the term $\mathbb E[\epsilon_n \epsilon_{n+1}]$ in \eqref{eq:DCE_oper}, we consider the properties of the encoder $f_{CE}$. 
%This is a minimum distance encoder with respect to a set of codewords drawn independently from a distribution that attains the DRF of $Y^{N_M}$ at bitrate $\bar{R} = MR$. 
The joint distribution of two sequence $Y^{N_M}$ and $\hat{Y}^{N_M}$, at the input and output of the encoder, respectively, behaves as if both sequences were drawn from the joint $P^*_{Y^{N_M},\hat{Y}^{N_M}}$ that attains the DRF of the vector $Y^{N_M}$ \cite{gallager1968information, kontoyiannis2006mismatched}. In our case, this distribution is defined by
\[
Y^{N_M} = \widehat Y^{N_M} + UZ^{N_M},
\]
where $U$ is the matrix of eigenvalues in the eigenvalue decomposition 
\begin{displaymath}
\begin{split} \label{eq:eigenvalues_decomp}
    \Sigma_Y = U \Lambda_Y U^H
\end{split}
\end{displaymath}
of $\Sigma_Y$, and $Z^{N_M} \sim \mathcal N(0,\Lambda_\theta)$ where $(\Lambda_\theta)_{n,n} = \min\{(\Lambda_Y)_{n,n}, \theta \}$. The rows of $U$ and entries of $\Sigma_Y$ are given by \cite{berger1970information}:
\begin{equation}
\begin{split}
    \lambda_k =& \sigma_Y^2 \left[2 \sin \left (\frac{2k-1}{2N_M+1} \frac{\pi}{2}\right) \right]^{-2} \\
    u_k[n] &= A_k \sin \left(\frac{2k-1}{2N_M+1} \pi n \right),
\end{split}
\end{equation}
where $k = 1, ..., N_M$, $\sigma_Y^2 = M$, and $A_k$ is a normalization constant. Let $\hat{\epsilon} \triangleq Y - \widehat{Y}$. From the above, we conclude 
\[
 \mathbb E \left[ \epsilon_n \epsilon_{n+1}\right] = \mathbb E \left[ \hat{\epsilon}\hat{\epsilon}_{n+1}\right]
 %= \left( U \Lambda_\theta U^H \right)_{n,n+1} 
 = \sum_{k=1}^{N_M} u_k[n] u_k[n+1]. 
\]
We showed in the Appendix C that 
\begin{equation}
    \label{eq:cross_term_convergence}
\frac{1}{N} \sum_{n=1}^{N_M-1} \mathbb E \left[ \hat{\epsilon}_n \hat{\epsilon}_{n+1} \right] \overset{N\rightarrow \infty}{\longrightarrow}  \int_0^1 \min \left\{ \theta', S(\phi) \right\} \cos(\pi \phi) d\phi,
\end{equation}
and $\mathbb E[\epsilon_{N_M}^2]/N \rightarrow 0$. This expression leads to the second term in \eqref{eq:CE_DRF_D}. \QEDA

\end{NoHyper}
\end{document}